\documentclass[a4paper]{jpconf}
\usepackage{graphicx}
\usepackage{amsmath,amsfonts,amsthm,amssymb}
\usepackage{booktabs} 
\usepackage{wasysym} 
\usepackage{bm}
\usepackage[usenames,dvipsnames]{color}
\usepackage{tikz-cd}
\usetikzlibrary{matrix,arrows,calc,positioning,shapes}
\usepackage{calc}
\usepackage{mathrsfs}
\usepackage{mathtools}
\usepackage{pdflscape}
\usepackage[colorlinks=true,linkcolor=blue,citecolor=red,urlcolor=Violet,bookmarksdepth=subsection,final]{hyperref}

\newtheorem{thm}{Theorem}[section]
\newtheorem{lem}[thm]{Lemma}

\newtheorem{dfn}{Definition}[section]

\makeatletter
\renewcommand*\env@matrix[1][*\c@MaxMatrixCols c]{%
  \hskip -\arraycolsep
  \let\@ifnextchar\new@ifnextchar
  \array{#1}}
\makeatother

\tikzset{ampersand replacement=\&}

\allowdisplaybreaks[1] 

  \newcommand{\oo}{\infty}
  
  \newcommand{\B}{\mathcal{B}}
\renewcommand{\d}{\mathrm{d}}
  \newcommand{\D}{\mathcal{D}}

  \newcommand{\coker}{\operatorname{coker}}
  \newcommand{\dalf}{{}^4\square}
  \newcommand{\del}{\partial}
  
  \newcommand{\gf}{{}^4g}
  \newcommand{\grf}{{}^4\nabla}
  
  \newcommand{\id}{\mathrm{id}}

  \newcommand{\sqf}{{}^4\square}

  \newcommand{\eps}{\varepsilon}

\makeatletter
\newcommand{\superimpose}[2]{%
  {\ooalign{$#1\@firstoftwo#2$\cr\hfil$#1\@secondoftwo#2$\hfil\cr}}}
\makeatother

\newsavebox{\sqzerbox}
\savebox{\sqzerbox}{$\mathpalette\superimpose{{\square}{\raisebox{.2ex}{\scalebox{.7}{0}}}}$}

\newsavebox{\sqonebox}
\savebox{\sqonebox}{$\mathpalette\superimpose{{\square}{\raisebox{.2ex}{\scalebox{.7}{1}}}}$}
\newcommand{\sqone}{\usebox{\sqonebox}}
\newsavebox{\sqtwobox}
\savebox{\sqtwobox}{$\mathpalette\superimpose{{\square}{\raisebox{.2ex}{\scalebox{.7}{2}}}}$}

\newsavebox{\cronebox}
\savebox{\cronebox}{$\mathpalette\superimpose{{\bigcirc}{\raisebox{.1ex}{\scalebox{.7}{1}}}}$}

\newsavebox{\crtwobox}
\savebox{\crtwobox}{$\mathpalette\superimpose{{\bigcirc}{\raisebox{.1ex}{\scalebox{.7}{2}}}}$}

\begin{document}
\title{Explicit triangular decoupling of the separated vector wave equation on Schwarzschild into scalar Regge-Wheeler equations}

\author{Igor Khavkine}

\address{Universit\`a di Milano, Via Cesare Saldini, 50, I-20133 Milano, Italy\\
Istituto Nazionale di Fisica Nucleare -- Sezione di Milano, Via Celoria, 16, I-20133 Milan, Italy}

\ead{igor.khavkine@unimi.it}

\begin{abstract}
We consider the vector wave equation on the Schwarzschild spacetime,
which can be considered as coming from the harmonic (or Lorenz) gauge
fixed Maxwell equations. After a separation of variables, the radial
mode equations form a complicated system of coupled linear ODEs. We
outline a precise abstract strategy to decouple this system into
triangular form, where the diagonal blocks consist of spin-$s$ scalar
Regge-Wheeler equations, with $s=0$ or $1$. This strategy is then
implemented to give an explicit transformation of the radial mode
equations (with nonzero frequency and angular momentum) into this
triangular form. Our decoupling goes a step further than previous
results in the literature by making the triangular form explicit and
reducing it as much as possible. Also, with the help of our abstractly
formulated decoupling strategy, we have significantly streamlined both
the presentation of the final results and the intermediate calculations.
Finally, we note that the vector wave equation is a simple model for
more complicated equations, like harmonic (or de~Donder) gauge fixed
linearized gravity, and backgrounds, like Kerr, where we expect the same
abstract decoupling strategy to work as well.
\end{abstract}

\section{Introduction}

The study of linear wave-like equations on a Schwarzschild black hole
background has a long history and many applications~\cite{chandrasekhar,
frolov-novikov}, both for scalar fields as well as higher rank tensor
fields. In this work, we will concentrate on the vector wave equation.
This equation, besides its purely mathematical significance as a model
for more general tensor wave equations, can be seen as the harmonic (or
Lorenz) gauge-fixed version of Maxwell equations~\cite[Sec.15.1]{ppv},
\cite{rosa-dolan}, or also as the residual gauge equation for harmonic
(or de~Donder) gauge-fixed linearized
gravity~\cite{fewster,bdm,berndtson}.

While the vector wave equation admits a complete separation of variables
on Schwarzschild, the resulting radial mode equations still represent a
complicated system of coupled ordinary differential equations (ODEs) in
the radial variable. What is worse is that, though this ODE system is
formally self-adjoint (it can be put into matrix Sturm-Liouville form),
it naturally defines a symmetric unbounded operator only on a Krein
space (a topological vector space with an indefinite scalar product).
This indefiniteness is directly traced to the indefiniteness of the
Lorentzian background metric and does not occur in the analogous
Riemannian problem. Thus, the abstract spectral theory of symmetric and
self-adjoint operators on a Hilbert space is not applicable in this
case, which significantly complicates important and fundamental
questions about this ODE system. For example: Can we prove or disproved
that its frequency spectrum is purely real? Can we prove that its
generalized eigenfunctions are complete? Can we construct a spectral
representation for the radial Green function?

To be able to answer these questions, one must appeal to some very
special structural properties of the vector wave equation. Namely, in
this work, we will show that the radial mode equations (at least for
modes with non-zero frequency and angular momentum) can be decoupled
into a triangular system, where the diagonal blocks are the spin-$s$
Regge-Wheeler equations, with $s=0$ or $1$. These Regge-Wheeler
equations are then of standard scalar Sturm-Liouville form, with very
well understood spectral properties. Such a decoupling has been
previously discussed in the literature in~\cite{berndtson} and
\cite{rosa-dolan}, though not in as explicit and conceptually clear
terms as we present below. In fact, the full explicit details of how the
original mode equations transform into the decoupled form and back are
not easy to extract from these references. The strategy that we follow
is basically the one of~\cite{berndtson}.%
	\footnote{It is by no means a simple task to extract this strategy
	from~\cite{berndtson}, as it only becomes apparent as the common
	pattern in the detailed and explicit calculations done for a sequence
	of examples of increasing complexity. Still, that work should be
	credited as (to our knowledge) the first to carry out this strategy
	explicitly for the example of the vector wave equation as well as the
	significantly more complicated Lichnerowicz equation, the latter
	corresponding to harmonic gauge fixed linearized gravity, on
	Schwarzschild.} %
Though, our results go a step further, by reducing the resulting
triangular form as much as possible.

In principle, starting from the triangular decoupled form, the spectral
properties of the Regge-Wheeler equations can completely determine the
spectral properties of the original radial mode equations. However, the
details of such a spectral analysis are left for future works.

In Section~\ref{sec:formal} we fix our notation, introduce the formal
properties of differential equations and operators, and formulate our
abstract decoupling strategy. In Section~\ref{sec:vw} we present a
complete separation of variables for the vector wave equation on
Schwarzschild spacetime and apply our strategy to give an explicit
triangular decoupling of the resulting radial mode equations into
spin-$0$ and spin-$1$ Regge-Wheeler scalars. In
Section~\ref{sec:discussion}, we discuss in more detail the novel
aspects of our results, as well as several avenues of further
investigation.

\section{Formal properties of differential equations and operators}\label{sec:formal}

\subsection{Morphisms and cochain maps between differential equations}

For the purposes of this work, a \emph{differential operator}, say $f$,
is always linear, with smooth coefficients, which we will write as
$f[u]$, where $u$ is a possibly vector valued function. Thus $f$ can
always be thought of as a matrix of scalar differential operators acting
on the components of $u$. By a \emph{scalar differential operator}, we
mean an operator that takes possibly vector valued functions into scalar
valued functions (corresponding to a matrix with a single row). Scalars
are real or complex numbers (for the sake of generality we will
generally allow complex scalars). Differential operators could be of any
order, including order zero, which just corresponds to multiplication by
some matrix valued function. While the operators can be thought of
as acting on smooth (vector valued) functions, we will be mostly
concerned with composition identities among operators, and so we will
not bother specifying precisely the domain or codomain of each operator.
This information can always be deduced from the context. Also, for the
abstract discussion below, it is also not necessary to fix the number of
independent variables, but we will be mostly concerned with applications
to ordinary differential operators (i.e.,\ acting on functions of a
single independent variable).

Below, we state some basic definitions and results concerning
differential equations and morphisms between them, which essentially
correspond to differential operators that map solutions to solutions.
The presentation is logically self-contained and does not extend far
beyond what is needed in the rest of the paper. However, for proper
context, these ideas can be seen as simple special cases of the more
general frameworks of $D$-modules~\cite[Sec.10.5]{seiler}, the category
of differential equations~\cite[Sec.VII.5]{vinogradov} or homological
algebra~\cite{weibel}. We will not delve into the precise connection
with these larger frameworks, but it is useful to mention that all
arrows/morphisms need to be reversed when our statements are interpreted
in the language of $D$-modules.

When referring to a morphism between differential equations, we
essentially mean a differential operator that maps solutions to
solutions. Sometimes given a differential equation $e[u] = 0$ we will
also refer to the operator $e$ as the equation. This should not lead to
any confusion.

\begin{dfn} \label{def:morphism}
Consider two differential equations $e[u] = 0$ and $\bar{e}[\bar{u}] =
0$.
\begin{enumerate}
\item[(a)]
	A differential operator $\bar{u} = k[u]$ is a \emph{morphism} from
	$e[u] = 0$ to $\bar{e}[\bar{u}] = 0$ when for any scalar differential
	operator of the form $\bar{p}\circ \bar{e}[\bar{u}]$ there exists a
	scalar differential operator of the form $q\circ e[u]$ such that
	\begin{equation}
		\bar{p}\circ\bar{e}[k[u]] = q\circ e[u] ,
	\end{equation}
	and hence the map $p \mapsto k_p$ given by the identity
	\begin{equation}
		p\circ k = k_p - h_p \circ e
	\end{equation}
	is well-defined on equivalence classes modulo $({\cdots})\circ \bar{e}$
	and $({\cdots})\circ e$ respectively.
\item[(b)]
	A pair of differential operators $\bar{u} = k[u]$ and $\bar{v} = g[v]$
	is a \emph{cochain map} from $e[u] = 0$ to $\bar{e}[\bar{u}] = 0$ when
	they satisfy the identity $\bar{e}\circ k = g\circ \bar{e}$.
\end{enumerate}
\end{dfn}
It is easy to check that a \emph{cochain map} is always a
\emph{morphism} and, as we will see shortly, a \emph{morphism} always
gives rise to a \emph{cochain map}, but not necessarily uniquely. The
two definitions are complementary to each other. Since a cochain map
comes with more structure, it is often more convenient for algebraic
manipulations. On the other hand, some properties of morphisms may be
easier to check. We illustrate morphisms and cochain maps respectively
as
\begin{equation}
\begin{tikzcd}[column sep=large,row sep=large]
	\bullet \ar{d}{e} \ar{r}{k} \&
	\bullet \ar[swap]{d}{\bar e} \\
	\bullet \ar[dashed]{r} \&
	\bullet
\end{tikzcd}
	\quad \text{and} \quad
\begin{tikzcd}[column sep=large,row sep=large]
	0 \ar{d} \& 0 \ar{d} \\
	\bullet \ar{d}{e} \ar{r}{k} \&
	\bullet \ar[swap]{d}{\bar e} \\
	\bullet \ar{d} \ar[swap]{r}{g} \&
	\bullet \ar{d} \\
	0 \& 0
\end{tikzcd} \, .
\end{equation}
The second diagram is meant to illustrate the terminology \emph{cochain
map}, which comes from homological algebra~\cite{weibel}. There, the primary
objects are \emph{cochain complexes}, which are sequences of linear
spaces (or more generally modules, or even just abelian groups)
connected by linear maps, of which any two successive ones compose to
zero. What is of interest about a cochain complex is its
\emph{cohomology} (the quotient of the kernel of one linear map by the
image of the preceding one). In our case the property of being a cochain
complex is trivially satisfied by the columns of the second diagram and
the cohomologies correspond to the spaces $\ker e$ and $\coker e$ of our
differential operator $e$. Morphisms between cochain complexes are
sequences of maps that commute with the linear maps of each complex. In
our case, this property is illustrated by the commutativity (the
identity $\bar{e}\circ k = g\circ e$) of the middle square of the second
diagram. It is straight forward to check that cochain maps induce
well-defined maps on cohomology, which in our case corresponds to $k$
mapping solutions of $e[u] = 0$ ($u\in \ker e$) to solutions of
$\bar{e}[\bar{u}] = 0$ ($u\in \ker\bar{e}$). In the sequel, we will
always be concerned with just the properties of the middle square of the
second diagram, hence the zero maps (illustrated above and below the
middle square) will from now on be omitted and left implicit.

\begin{lem} \label{lem:morphism-induces-map}
A morphism $k$ from $e[u] = 0$ to $\bar{e}[\bar{u}] = 0$ always induces a
cochain map $k,g$.
\end{lem}
\begin{proof}
Choosing $\bar{p} = \id$ in the definition of a morphism
(Definition~\ref{def:morphism}(a)), implies that there exists an
operator $g$ such that $\bar{e}\circ k = g\circ e$. The last identity
also implies that $k,g$ is a cochain map
(Definition~\ref{def:morphism}(b)).
\end{proof}
Of course in general the operator $g$ completing the morphism $k$ to a
cochain map may not be unique. The next two definitions introduce some
parallel properties of morphisms and cochain maps.

\begin{dfn}
A morphism $k$ from $e[u] = 0$ to $\bar{e}[\bar{u}] = 0$ is called
\emph{on-shell injective (surjective)} if the map $p \mapsto k_p$ is
surjective (injective) from equivalence classes modulo $({\cdots})\circ
\bar{e}$ to equivalence classes modulo $({\cdots})\circ e$. We say
\emph{on-shell bijective} to mean both on-shell injective and
surjective. We say that $k$ is \emph{on-shell vanishing} when $p \circ
k$ belongs to the same equivalence class as $0$ modulo $({\cdots})\circ
e$ for any operator $p$.
\end{dfn}

\begin{dfn}
A cochain map $k,g$ from $e[u] = 0$ to $\bar{e}[\bar{u}] = 0$ is said to
be \emph{induced by a homotopy} if $k = h\circ e$ and $g = \bar{e}\circ
h$ for some operator $h$ referred to as a \emph{(cochain) homotopy}:
\begin{equation}
\begin{tikzcd}[column sep=large,row sep=large]
	\bullet \ar[swap]{d}{e} \ar{r}{k=h\circ e} \&
	\bullet \ar{d}{\bar{e}} \\
	\bullet \ar[swap]{r}{g=\bar{e}\circ h} \ar[dashed]{ur}{h} \&
	\bullet
\end{tikzcd} \, .
\end{equation}
We say that $k,g$ has a \emph{left (right) inverse up to homotopy} if
there exists a cochain map $\bar{k},\bar{g}$ from $\bar{e}[\bar{u}] = 0$
to $e[u] = 0$ such that
\begin{equation*}
	\begin{aligned}
		\bar{k} \circ k &= \id - h\circ e \\
		\bar{g} \circ g &= \id - e\circ h
	\end{aligned}
	\quad
	\left(\begin{aligned}
		k\circ \bar{k} &= \id - \bar{h}\circ \bar{e} \\
		g\circ \bar{g} &= \id - \bar{e}\circ \bar{h}
	\end{aligned}\right)
\end{equation*}
for some homotopy $h$ of $e[u] = 0$ ($\bar{h}$ of $\bar{e}[\bar{u}] =
0$) into itself. A cochain map $k,g$ that has a two-sided inverse up to
homotopy is called an \emph{equivalence up to homotopy}.
\end{dfn}

The last two sets of definitions are related but not exactly equivalent.
In general, the existence of a left or right inverse is a stronger
property than being on-shell injective or surjective. Though, under
certain hypotheses, the link could be made stronger.
\begin{dfn} \label{def:determined}
Consider a differential equation $e[u] = 0$. Any operator $p$ such that
$e\circ p = 0$ is called a \emph{gauge symmetry (generator)} of $e$. Any
operator $q$ such that $q\circ e = 0$ is called a \emph{Noether
identity} (or \emph{compatibility operator}) of $e$. The equation is
said to be \emph{under-determined} if it has a non-zero gauge symmetry
and \emph{over-determined} if it has a non-zero Noether identity. The
equation is said to be \emph{over-under-determined} if it has both
non-zero gauge symmetries and Noether identities and \emph{determined}
(or \emph{normal}) if it has neither.
\end{dfn}

\begin{lem} \label{lem:morphisms-maps}
Consider equations $e[u] = 0$ and $\bar{e}[\bar{u}] = 0$.
\begin{enumerate}
\item[(a)]
	When $e$ is not over-determined, a morphism $\bar{u} = k[u]$ induces a
	\emph{unique} cochain map $k,g$.
\item[(b)]
	A cochain map $k,g$ has a left (right) inverse up to homotopy $\implies$
	$k$ is on-shell injective (surjective). If $k,g$ is induced by a
	homotopy, then $k$ is on-shell vanishing.
\item[(c)]
	When a morphism $\bar{u} = k[u]$ is on-shell bijective, it has a
	two-sided inverse morphism $u = \bar{k}[\bar{u}]$.
\item[(d)]
	When neither $e$ nor $\bar{e}$ is over-determined, an on-shell
	bijective morphism $\bar{u} = k[u]$ induces an equivalence up to
	homotopy. More precisely, it induces a cochain map $k,g$ (with $g$
	unique) and a two-sided inverse up to homotopy $\bar{k},\bar{g}$.
\end{enumerate}
\end{lem}
\begin{proof}
\begin{itemize}
\item[(a)]
	From Lemma~\ref{lem:morphism-induces-map}, we already know that there
	exists at least one operator $g$ such that $k,g$ is a cochain map. Let
	$g'$ be another such operator. Then $(g' - g)\circ e = \bar{e}\circ
	(k-k) = 0$. Hence, if $e$ has no non-vanishing Noether identities
	(Definition~\ref{def:determined}), we must have $g' = g$.
\item[(b)]
	Suppose that $\bar{k},\bar{g}$ is a left inverse to $k,g$ up to
	homotopy. Then, for any operator $q[u]$, the operator $\bar{p} = q\circ
	\bar{k}$ satisfies
	\begin{equation}
		\bar{p} \circ k
		= q \circ (\bar{k} \circ k)
		= q - (q\circ h) \circ e \, .
	\end{equation}
	Hence $k$ is on-shell injective.

	On the other hand, suppose that
	$\bar{k},\bar{g}$ is a right inverse to $k,g$ up to homotopy. Then any
	operator $\bar{p}$ such that $\bar{p}\circ k = q\circ e$, we have the
	identity
	\begin{equation}
		p \circ (\id - \bar{h} \circ \bar{e})
		= (p\circ k)\circ \bar{k}
		= q\circ (e \circ \bar{k})
		= (q\circ \bar{g}) \circ \bar{e} \, .
	\end{equation}
	Hence we must have $\bar{p} = (q\circ\bar{g} + p\circ h)\circ \bar{e}$
	and therefore $k$ is on-shell surjective.

	Finally if $k = h\circ e$ and $g = \bar{e}\circ h$, then for any
	operator $p$ we have $p\circ k = (p\circ h) \circ e$, meaning that $k$
	is on-shell vanishing.
\item[(c)]
	By on-shell injectivity, there must exist operators $\bar{k}$ and $h$
	such that $\bar{k} \circ k = \id - h\circ e$. Then
	\begin{equation}
		(\id - k\circ \bar{k}) \circ k
		= k\circ (\id - \bar{k}\circ k)
		= (k\circ h) \circ e \, .
	\end{equation}
	Hence, by on-shell surjectivity, there must exist an operator
	$\bar{h}$ such that $k\circ \bar{k} = \id - \bar{h} \circ \bar{e}$.
	Now, it remains only to check that $\bar{k}$ is actually a morphism.
	For that, note the identity
	\begin{equation}
		(e\circ \bar{k}) \circ k
		= e\circ (\bar{k}\circ k)
		= e - e\circ h \circ e
		= (\id - e\circ h) \circ e \, .
	\end{equation}
	But by on-shell surjectivity of $k$ this means that there exists an
	operator $\bar{g}$ such that $e\circ \bar{k} = \bar{g} \circ \bar{e}$,
	meaning that $\bar{k},\bar{g}$ is a cochain map and hence $\bar{k}$
	itself is a morphism that is a two-sided on-shell inverse to $k$.
\item[(d)]
	By part (c), we know that an on-shell bijective morphism $k$ has a two
	sided inverse morphism $\bar{k}$, satisfying
	\begin{equation}
		\bar{k} \circ k = \id - h\circ e
		\quad \text{and} \quad
		k\circ\bar{k} = \id - \bar{h}\circ \bar{e} \, ,
	\end{equation}
	for some operators $h$ and $\bar{h}$. By part (a), they can both be
	uniquely completed to cochain maps, $k,g$ and $\bar{k},\bar{g}$. It
	remains to check the relation between $g$ and $\bar{g}$. Consider the
	identities
	\begin{align}
		(\bar{g}\circ g) \circ e
		&= \bar{g} \circ (g\circ e)
		= (\bar{g} \circ \bar{e}) \circ k
		= e \circ (\bar{k} \circ k)
		= (\id - e\circ h) \circ e  \, , \\
		(g\circ \bar{g}) \circ \bar{e}
		&= g \circ (\bar{g}\circ \bar{e})
		= (g \circ e) \circ \bar{k}
		= \bar{e} \circ (k \circ \bar{k})
		= (\id - \bar{e}\circ \bar{h}) \circ \bar{e}  \, .
	\end{align}
	Since neither $e$ nor $\bar{e}$ has non-vanishing Noether identities,
	we must conclude that
	\begin{equation}
		\bar{g} \circ g = \id - e\circ h
		\quad \text{and} \quad
		g\circ\bar{g} = \id - \bar{e}\circ \bar{h} \, .
	\end{equation}
	Hence, the cochain map $\bar{k},\bar{g}$ is a two-sided to $k,g$ up to
	homotopy.
\end{itemize}
\end{proof}

\subsection{Block-triangular decoupling}

In this section, we show some equations can be transformed (or
\emph{decoupled}) into block triangular form. These results will be the
building blocks of our abstract decoupling strategy, which will be
discussed in Section~\ref{sec:decoupling}. Note that we do not state
these results in their most general form, but only in sufficient
generality to be useful later on.

\begin{lem} \label{lem:induce-triang}
Consider the differential equations $e[u] = 0$ and $f[v] = 0$. A
morphism $v = k[u]$ induces an
on-shell bijective morphism from $e[u] = 0$ to $\bar{e}[\bar{u}] = 0$,
with $\bar{u} = [\begin{smallmatrix} u \\ v \end{smallmatrix}]$ and
$\bar{e}$ in upper block-triangular form, namely
\begin{equation}
\begin{tikzcd}[column sep=large,row sep=large]
	\bullet \ar[swap]{d}{e} \ar{r}{k} \&
	\bullet \ar{d}{f} \\
	\bullet \ar[swap]{r}{g} \&
	\bullet
\end{tikzcd}
	\quad \implies \quad
\begin{tikzcd}[column sep=large,row sep=large]
	\bullet \ar[swap]{d}{e} \ar{r}{k
		= \begin{bmatrix}
			\id \\
			\cmidrule(lr){1-1}
			k
		\end{bmatrix}} \&
	\bullet \ar{d}{\bar e
		= \begin{bmatrix}[c|c]
			e & 0 \\
			k & -\id \\
			\cmidrule(lr){1-2}
			0 & f
		\end{bmatrix}} \\
	\bullet \ar[swap]{r}{\begin{bmatrix}
			\id \\
			0 \\
			\cmidrule(lr){1-1}
			g
		\end{bmatrix}} \&
	\bullet
\end{tikzcd}
	\quad \iff \quad
\begin{tikzcd}[column sep=large,row sep=large]
	\bullet
		\ar[swap]{d}{\begin{bmatrix}[c|c]
			e & 0 \\
			k & -\id \\
			\cmidrule(lr){1-2}
			0 & f
		\end{bmatrix}}
		\ar{r}{\begin{bmatrix}[c|c] \id & 0 \end{bmatrix}}
	\&
	\bullet
		\ar{d}{e}
	\\
	\bullet
		\ar[swap]{r}{\begin{bmatrix}[cc|c] \id & 0 & g \end{bmatrix}}
	\&
	\bullet
\end{tikzcd} ,
\end{equation}
where the last two morphisms are mutually on-shell inverse.
\end{lem}

We say that the morphism $k$ \emph{decouples} $f$ from $e$. Once $f$ has
been decoupled into a lower diagonal block, any transformations of
remaining upper diagonal block transform the block-triangular form as
follows.

\begin{lem} \label{lem:top-triang-inv}
Consider the equations $e[u] = 0$, $\bar{e}[\bar{u}] = 0$, $f[v] = 0$,
and
\begin{equation}
	\begin{bmatrix}
		e & \Delta \\
		0 & f
	\end{bmatrix}
	\begin{bmatrix}
		u \\
		v
	\end{bmatrix}
	= 0 \, ,
\end{equation}
where the operator $\Delta$ has the property that for any Noether
identity $n\circ e = 0$ we have $n\circ \Delta = m\circ f$ for some
operator $m$. An on-shell bijective morphism $k$ between $e$ and
$\bar{e}$, which more precisely satisfies the identities (which are
slightly weaker than equivalence up to homotopy)
\begin{equation}
\begin{tikzcd}[column sep=large,row sep=large]
	\bullet \ar{d}{e} \ar[shift left]{r}{k} \&
	\bullet \ar[swap]{d}{\bar e} \ar[shift left]{l}{\bar k}
	\\
	\bullet \ar[shift left]{r}{g} \&
	\bullet \ar[shift left]{l}{\bar g}
\end{tikzcd} \, ,
	\qquad
	\begin{aligned}
		\bar{k} \circ k &= \id - h\circ e \, , \\
		\bar{g} \circ g &= \id - e\circ h - n \, .
	\end{aligned}
\end{equation}
with a Noether identity $n\circ e = 0$, induces the following mutually
on-shell bijective morphisms of equations in triangular form, where
$\bar{\Delta} = g\circ \Delta$:
\begin{equation}
\begin{tikzcd}[column sep=large,row sep=large]
	\bullet
		\ar[swap]{d}{\begin{bmatrix}
			e & \Delta \\
			0 & f
		\end{bmatrix}}
		\ar{r}{\begin{bmatrix} k & 0 \\ 0 & \id \end{bmatrix}}
	\&
	\bullet
		\ar{d}{\begin{bmatrix}
			\bar{e} & \bar{\Delta} \\
			0 & f
		\end{bmatrix}}
	\\
	\bullet
		\ar[swap]{r}{\begin{bmatrix} g & 0 \\ 0 & \id \end{bmatrix}}
	\&
	\bullet
\end{tikzcd}
	\quad \iff \quad
\begin{tikzcd}[column sep=large,row sep=large]
	\bullet
		\ar[swap]{d}{\begin{bmatrix}
			\bar{e} & \bar{\Delta} \\
			0 & f
		\end{bmatrix}}
		\ar{r}{\begin{bmatrix} \bar{k} & -h\circ\Delta \\ 0 & \id \end{bmatrix}}
	\&
	\bullet
		\ar{d}{\begin{bmatrix}
			e & \Delta \\
			0 & f
		\end{bmatrix}}
	\\
	\bullet
		\ar[swap]{r}{\begin{bmatrix} \bar{g} & m \\ 0 & \id \end{bmatrix}}
	\&
	\bullet
\end{tikzcd} \, .
\end{equation}
\end{lem}
Of course, we should note that the block-triangular equations obtained
in the earlier Lemma~\ref{lem:induce-triang} satisfy the hypotheses of
Lemma~\ref{lem:top-triang-inv}.
\begin{proof}
Consider the following identities:
\begin{equation}
	\begin{bmatrix}
		\bar{e} & \bar{\Delta} \\
		0 & f
	\end{bmatrix}
	\circ
	\begin{bmatrix}
		k & 0 \\
		0 & \id
	\end{bmatrix}
	-
	\begin{bmatrix}
		g & 0 \\
		0 & \id
	\end{bmatrix}
	\circ
	\begin{bmatrix}
		e & \Delta \\
		0 & f
	\end{bmatrix}
	=
	\begin{bmatrix}
		\bar{e}\circ k - g\circ e & \bar{\Delta} - g\circ \Delta \\
		0 & f - f
	\end{bmatrix} \, ,
\end{equation}
\begin{multline}
	\begin{bmatrix}
		e & \Delta \\
		0 & f
	\end{bmatrix}
	\circ
	\begin{bmatrix}
		\bar{k} & -h\circ\Delta \\
		0 & \id
	\end{bmatrix}
	-
	\begin{bmatrix}
		\bar{g} & m \\
		0 & \id
	\end{bmatrix}
	\circ
	\begin{bmatrix}
		\bar{e} & g\circ\Delta \\
		0 & f
	\end{bmatrix}
	\\
	=
	\begin{bmatrix}
		e\circ\bar{k}-\bar{g}\circ \bar{e} &
			(\id - e\circ h - n - \bar{g}\circ g)\circ\Delta
				+ (n\circ\Delta - m\circ f)
		\\
		0 & f - f
	\end{bmatrix} \, .
\end{multline}
They show that the diagrams in the Lemma show actual cochain maps. It is
also straight forward to see that they are mutually inverse up to
homotopy.
\end{proof}

We also have the analogous
\begin{lem} \label{lem:bot-triang-inv}
Consider equations $e[u] = 0$, $\bar{e}[\bar{u}] = 0$, $f[v] = 0$ and
\begin{equation}
	\begin{bmatrix}
		e & \Delta \\
		0 & f
	\end{bmatrix}
	\begin{bmatrix}
		u \\
		v
	\end{bmatrix}
	= 0 \, ,
\end{equation}
An equivalence up to homotopy $k,g$ between $e$ and $\bar{e}$
\begin{equation}
\begin{tikzcd}[column sep=large,row sep=large]
	\bullet \ar{d}{e} \ar[shift left]{r}{k} \&
	\bullet \ar[swap]{d}{\bar e} \ar[shift left]{l}{\bar k}
	\\
	\bullet \ar[shift left]{r}{g} \&
	\bullet \ar[shift left]{l}{\bar g}
\end{tikzcd} \, ,
	\qquad
	\begin{aligned}
		\bar{k} \circ k &= \id - h\circ e \, , \\
		\bar{g} \circ g &= \id - e\circ h \, .
	\end{aligned}
\end{equation}
induces the following mutually on-shell bijective morphisms of equations
in triangular form, where $\Delta = \bar{\Delta}\circ k$:
\begin{equation}
\begin{tikzcd}[column sep=large,row sep=large]
	\bullet
		\ar[swap]{d}{\begin{bmatrix}
			f & \Delta \\
			0 & e
		\end{bmatrix}}
		\ar{r}{\begin{bmatrix} \id & 0 \\ 0 & k \end{bmatrix}}
	\&
	\bullet
		\ar{d}{\begin{bmatrix}
			f & \bar{\Delta} \\
			0 & \bar{e}
		\end{bmatrix}}
	\\
	\bullet
		\ar[swap]{r}{\begin{bmatrix} \id & 0 \\ 0 & g \end{bmatrix}}
	\&
	\bullet
\end{tikzcd}
	\quad \iff \quad
\begin{tikzcd}[column sep=large,row sep=large]
	\bullet
		\ar[swap]{d}{\begin{bmatrix}
			f & \bar{\Delta} \\
			0 & \bar{e}
		\end{bmatrix}}
		\ar{r}{\begin{bmatrix} \id & 0 \\ 0 & \bar{k} \end{bmatrix}}
	\&
	\bullet
		\ar{d}{\begin{bmatrix}
			f & \Delta \\
			0 & e
		\end{bmatrix}}
	\\
	\bullet
		\ar[swap]{r}{\begin{bmatrix} \id & -\bar{\Delta}\circ h \\ 0 & \bar{g} \end{bmatrix}}
	\&
	\bullet
\end{tikzcd}
\end{equation}
\end{lem}
\begin{proof}
Consider the following identities:
\begin{equation}
	\begin{bmatrix}
		f & \bar{\Delta} \\
		0 & \bar{e}
	\end{bmatrix}
	\circ
	\begin{bmatrix}
		\id & 0 \\
		0 & k
	\end{bmatrix}
	-
	\begin{bmatrix}
		\id & 0 \\
		0 & g
	\end{bmatrix}
	\circ
	\begin{bmatrix}
		f & \Delta \\
		0 & e
	\end{bmatrix}
	=
	\begin{bmatrix}
		f - f & \bar{\Delta}\circ k - \Delta \\
		0 & \bar{e}\circ k - g\circ e
	\end{bmatrix} \, ,
\end{equation}
\begin{equation}
	\begin{bmatrix}
		f & \Delta \\
		0 & e
	\end{bmatrix}
	\circ
	\begin{bmatrix}
		\id & 0 \\
		0 & \bar{k}
	\end{bmatrix}
	-
	\begin{bmatrix}
		\id & -\bar{\Delta}\circ h \\
		0 & \bar{g}
	\end{bmatrix}
	\circ
	\begin{bmatrix}
		f & \bar{\Delta} \\
		0 & \bar{e}
	\end{bmatrix}
	=
	\begin{bmatrix}
		f - f &
			\bar{\Delta}\circ (k \circ\bar{k} - \id + \bar{h}\circ \bar{e})
		\\
		0 & e\circ\bar{k}-\bar{g}\circ \bar{e}
	\end{bmatrix} \, .
\end{equation}
They show that the diagrams in the Lemma show actual cochain maps. It is
also straight forward to see that they are mutually inverse up to
homotopy.
\end{proof}

\subsection{Reducing the triangular form}\label{sec:triang-reduce}

Given an equation that is in block-\emph{triangular} form, when is it
equivalent to an equation in block-\emph{diagonal} form? That is, when
can an equation in triangular form be further reduced?

In the sequel, we will only need the following result.
\begin{lem}
Consider the mutually inverse up to homotopy cochain maps
\begin{equation}
\begin{tikzcd}[column sep=large,row sep=large]
	\bullet
		\ar[swap]{d}{\begin{bmatrix} e & f \\ 0 & \bar{e} \end{bmatrix}}
		\ar{r}{\begin{bmatrix} \id & \delta \\ 0 & \id \end{bmatrix}}
	\&
	\bullet
		\ar{d}{\begin{bmatrix} e & 0 \\ 0 & \bar{e} \end{bmatrix}}
	\\
	\bullet
		\ar[swap]{r}{\begin{bmatrix} \id & \eps \\ 0 & \id \end{bmatrix}}
	\&
	\bullet
\end{tikzcd}
	\quad \iff \quad
\begin{tikzcd}[column sep=large,row sep=large]
	\bullet
		\ar[swap]{d}{\begin{bmatrix} e & 0 \\ 0 & \bar{e} \end{bmatrix}}
		\ar{r}{\begin{bmatrix} \id & -\delta \\ 0 & \id \end{bmatrix}}
	\&
	\bullet
		\ar{d}{\begin{bmatrix} e & f \\ 0 & \bar{e} \end{bmatrix}}
	\\
	\bullet
		\ar[swap]{r}{\begin{bmatrix} \id & -\eps \\ 0 & \id \end{bmatrix}}
	\&
	\bullet
\end{tikzcd} \, .
\end{equation}
The two squares in these diagrams actually commute (that is, we can kill
the off-diagonal block $f$) if and only if the operators $\delta$ and
$\eps$ satisfy the identity
\begin{equation}
	e \circ \delta = f + \eps \circ \bar{e} \, .
\end{equation}
\end{lem}
Essentially, the operators $\delta$ and $\eps$ allow us to solve the
equation is $e[u] = f[\bar{u}]$ by $u = \delta[\bar{u}]$ whenever
$\bar{e}[\bar{u}] = 0$.
\begin{proof}
The desired result follows immediately from the identity
\begin{equation}
	\begin{bmatrix} e & 0 \\ 0 & \bar{e} \end{bmatrix}
	\circ
	\begin{bmatrix} \id & \delta \\ 0 & \id \end{bmatrix}
	-
	\begin{bmatrix} \id & \eps \\ 0 & \id \end{bmatrix}
	\circ
	\begin{bmatrix} e & f \\ 0 & \bar{e} \end{bmatrix}
	=
	\begin{bmatrix}
		e - e & e\circ\delta - f - \eps\circ\bar{e} \\
		0 & \bar{e} - \bar{e}
	\end{bmatrix} \, .
\end{equation}
The reverse direction is completely analogous.
\end{proof}

\subsection{Formal adjoints} \label{sec:adjoints}

Recall that we are only dealing with linear ordinary differential
operators, with derivative operator $\del_r$, where each operator can be
written as a matrix of scalar differential operators with smooth
coefficients. Given a matrix linear differential operator $k$, we define
its formal adjoint $k^*$ using the following rules. If $k$ is a zero-th
order scalar operator acting on a single scalar argument, basically
multiplication by some scalar function $k(r)$, then $k^*$ is
multiplication by the complex conjugate function $k(r)^*$.
If $k^* = \del_r$ acting on a single scalar argument, then $k^* =
-\del_r$. If $k$ and $g$ are scalar operators acting on a single scalar
argument, then $(k+g)^* = k^*+g^*$ and $(k\circ g)^* = g^*\circ k^*$.
Finally, if $k = [k_{ij}]$ is a matrix of scalar differential operators,
then $(k^*)_{ij} = k_{ji}^*$. It is straight forward to check that this
gives rise to a consistent and unique definition, which satisfies
$(k\circ g)^* = g^*\circ k^*$ and $k^** = k$, now with any composable
operators $k$ and $g$, and agrees with the common notion of formal
adjoint. When dealing with a family of operators $k_\omega$ parametrized
by a possibly complex parameter $\omega$, we adopt the convension that
the family of adjoint operators is parametrized as $(k^*)_\omega =
(k_{\omega^*})^*$. With this convention, $(k^*)_\omega = (k_\omega)^*$
for real values $\omega = \omega^*$ of the parameter.
An operator satisfying $k^* = k$ is called \emph{(formally) self-adjoint}.

Sometimes, it will be convenient for us, when introducing a new
differential operator, to write it directly as $k^*$, even if we have
not explicitly introduced the operator $k$. This is possible, since any
differential operator is the formal adjoint of some differential
operator. This convention allows us to highlight some special features
of the equations and differential operators that we will be working
with, and also reigns in a bit the proliferation of new notation.


\subsection{Gauge fixing, gauge-invariant fields and triangular decoupling}\label{sec:decoupling}

In this section, we give an abstract overview of the triangular
decoupling procedure that will be illustrated on a specific example in
Section~\ref{sec:vw}. The abstract approach allows us to separate the general
strategy from the calculational details of a specific example, so that
it could be applied in other examples as well. One benefit of having a
general strategy is that it eliminates much trial-and-error and
unnecessary calculations in more complex examples.

We start with a system of differential equations $E[u] = 0$ and aim to
find explicit equivalence of this system to one that is in block
upper-triangular form, with as many off-diagonal terms set to zero (or
\emph{reduced}) as possible. Such a \emph{triangular decoupling} will be
possible because of the large amount of structure that we can assume
about our equation. Essentially, we expect $E$ to be the gauge-fixed
form of a more fundamental self-adjoint equation $E_0[u] = 0$ that has
both gauge symmetries and (hence necessarily) Noether identities. The
gauge fixing and the modifications of $E$ from $E_0$ are chosen such
that $E$ remains self-adjoint and it induces a nice residual gauge
equation. We will refer to such a gauge fixing condition as
\emph{harmonic}, because harmonic or generalized harmonic gauges tend to
satisfy all of these requirements.

While we do try to keep track of self-adjointness of our equations,
which is helpful at least as an error-check on explicit calculations,
adjoint operators do not play a specific role in the decoupling strategy
outlined below. Hence, it would work just as well if the operators that
appear as adjoints are replaced by independent operators, which satisfy
all the same composition identities.

The general strategy will be to show that the original dependent
variables separate into \emph{gauge modes}, \emph{gauge invariant modes}
and \emph{constraint violating modes}, which are coupled to each other in
a hierarchical (hence \emph{triangular}) and minimal way.

Let us introduce the following differential operators and identities
among them:
\begin{itemize}
\item
	$E_0$---self-adjoint equation with gauge symmetry,
\item
	$D$---gauge symmetry generator, $E_0 \circ D = 0$,
\item
	$D^*$---Noether identity, $D^*\circ E_0 = 0$,
\item
	$T$---harmonic gauge fixing condition, $T\circ D = R\circ \D_D$,
\item
	$\D_D$---self-adjoint residual gauge equation,
\item
	$R$---self-adjoint gauge fixing correction, $R^* = R$,
\item
	$E=E_0+T^*\circ \bar{R} \circ T$---self-adjoint gauge-fixed equation,
	$\bar{R} = R^{-1}$,
\item
	$\Phi_0$---gauge invariant field, $\Phi_0\circ D = 0$,
\item
	$\D_\Phi$---self-adjoint dynamical equation for gauge-invariant
	fields, $\D_\Phi \circ \Phi_0 = \bar{\Phi}_0^* \circ E_0$, with
	$E\circ D = T^*\circ \D_D$ and $\D_D^* \circ T = D^*\circ E$.
\end{itemize}
The operators
\begin{equation}
	D[\phi_0] \, ,
	\quad
	\Phi_0 \, ,
	\quad
	T
\end{equation}
respectively separate out, in a sense that will become clear by the end
of this section, the \emph{gauge modes}, the \emph{gauge invariant
modes} and the \emph{constraint violating modes}. 
The general strategy will be to show that these modes couple
to each other only in a hierarchical (hence \emph{triangular}) way, and
to find slight modifications to them that reduce this coupling as much
as possible. The subscript $s$ on $\phi_s$ and $\psi_s$ indicates the
spin of the corresponding Regge-Wheeler operators $\D_s$ that appear on
the diagonal of the decoupled form.

The above operators fit into the following commutative diagrams:
\begin{equation}
\begin{tikzcd}[column sep=large,row sep=large]
	\bullet \ar[swap]{d}{\D_D} \ar{r}{D} \&
	\bullet \ar{d}{E}
	\\
	\bullet \ar[swap]{r}{T^*} \&
	\bullet
\end{tikzcd} \, ,
	\quad
\begin{tikzcd}[column sep=large,row sep=large]
	\bullet \ar[swap]{d}{E} \ar{r}{T} \&
	\bullet \ar{d}{\D_D^*}
	\\
	\bullet \ar[swap]{r}{D^*} \&
	\bullet
\end{tikzcd} \, ,
	\quad
\begin{tikzcd}[column sep=large,row sep=large]
	\bullet \ar[swap]{d}{E_0} \ar{r}{\Phi_0} \&
	\bullet \ar{d}{\D_\Phi}
	\\
	\bullet \ar[swap]{r}{\bar{\Phi}_0^*} \&
	\bullet
\end{tikzcd} \, .
\end{equation}
Even though we have $\D_D^* = \D_D$, we use this notation to formally
distinguish between the dynamical equation satisfied by the
longitudinal/constraint-violating modes $T$ and the residual gauge modes
$D$. Note that the operator $\bar{\Phi}_0^*$ is not uniquely fixed by
completing the commutative diagram between $E_0$, $\Phi_0$ and $\D_\Phi$.
If $\bar{\Phi}_1^* = \bar{\Phi}_0^* + N$ is any other such operator, then
$N \circ E_0 = 0$ and hence $N$ is a Noether identity for
$E_0$.

As a first step, we use the operators $T$ and $D^*$ and
Lemma~\ref{lem:induce-triang} to decouple $\D_D^*$ from $E$. We have the
following morphism from $E$ to a triangular form:
\begin{equation} \label{eq:E-T-triang}
\begin{tikzcd}[column sep=huge,row sep=huge]
	\bullet
		\ar[swap]{d}{E}
		\ar{r}{\begin{bmatrix} \id \\ T \end{bmatrix}}
	\&
	\bullet
		\ar{d}{\begin{bmatrix}
			E & 0 \\
			T & -\id \\
			0 & \D_D^*
		\end{bmatrix}}
	\\
	\bullet
		\ar[swap]{r}{\begin{bmatrix} \id \\ 0 \\ D^* \end{bmatrix}}
	\&
	\bullet
\end{tikzcd}
	\quad \iff \quad
\begin{tikzcd}[column sep=huge,row sep=huge]
	\bullet
		\ar[swap]{d}{\begin{bmatrix}
			E & 0 \\
			T & -\id \\
			0 & \D_D^*
		\end{bmatrix}}
		\ar{r}{\begin{bmatrix} \id & 0 \end{bmatrix}}
	\&
	\bullet
		\ar{d}{E}
	\\
	\bullet
		\ar[swap]{r}{\begin{bmatrix} \id & 0 & 0 \end{bmatrix}}
	\&
	\bullet
\end{tikzcd} \, ,
\end{equation}
where the last two morphisms are on-shell mutually inverse.

Next, we would like to use the operators $\Phi_0$ and $\bar{\Phi}_0^*$ to
decouple $\D_\Phi$ from $[\begin{smallmatrix} E \\ T
\end{smallmatrix}]$. Let us introduce the operator
\begin{equation}
	\bar{\Delta}_\Phi = \bar{\Phi}_0^*\circ T^* \circ \bar{R} \, ,
\end{equation}
which is necessary to decouple $\D_\Phi$, since it completes the
commutative diagram
\begin{equation}
\begin{tikzcd}[column sep=large,row sep=large]
	\bullet \ar[swap]{d}{\begin{bmatrix} E \\ T \end{bmatrix}} \ar{r}{\Phi_0} \&
	\bullet \ar{d}{\D_\Phi}
	\\
	\bullet \ar[swap]{r}{\begin{bmatrix} \bar{\Phi}_0^* & -\bar{\Delta}_\Phi \end{bmatrix}} \&
	\bullet
\end{tikzcd} \, .
\end{equation}
Once again, applying Lemma~\ref{lem:induce-triang} we get the mutually
on-shell inverse morphisms
\begin{equation} \label{eq:E-T-Phi0-triang}
\begin{tikzcd}[column sep=huge,row sep=huge]
	\bullet
		\ar[swap]{d}{\begin{bmatrix} E \\ T \end{bmatrix}}
		\ar{r}{\begin{bmatrix} \id \\ \Phi_0 \end{bmatrix}}
	\&
	\bullet
		\ar{d}{\begin{bmatrix}
			E & 0 \\
			T & 0 \\
			\Phi_0 & -\id \\
			0 & \D_\Phi
		\end{bmatrix}}
	\\
	\bullet
		\ar[swap]{r}{\begin{bmatrix}
			\id & 0 \\
			0 & \id \\
			0 & 0 \\
			\bar{\Phi}_0^* & -\bar{\Delta}_\Phi
		\end{bmatrix}}
	\&
	\bullet
\end{tikzcd}
	\iff 
\begin{tikzcd}[column sep=huge,row sep=huge]
	\bullet
		\ar[swap]{d}{\begin{bmatrix}
			E & 0 \\
			T & 0 \\
			\Phi_0 & -\id \\
			0 & \D_\Phi
		\end{bmatrix}}
		\ar{r}{\begin{bmatrix} \id & 0 \end{bmatrix}}
	\&
	\bullet
		\ar{d}{\begin{bmatrix} E \\ T \end{bmatrix}}
	\\
	\bullet
		\ar[swap]{r}{\begin{bmatrix}
			\id & 0 & 0 & 0 \\
			0 & \id & 0 & 0
		\end{bmatrix}}
	\&
	\bullet
\end{tikzcd} \, .
\end{equation}
The operator $\bar{\Delta}_\Phi$ is an obstacle to decoupling $\D_\Phi$
directly from $E$. This obstacle becomes trivial when $\bar{\Delta}_\Phi =
\D_\Phi\circ q$ for some operator $q$, since then replacing $\Phi_0$
with $\Phi_0 + q^*\circ T$ and setting $\bar{\Delta}_\Phi = 0$ does the
required job. If we had used $\Phi_0 + q\circ T$ from the start, then
the final triangular form that we will arrive at will not have any
coupling between $T$ and $\Phi_0 + q\circ T$. We will actually be able
to achieve such a reduced triangular form, by using the strategy from
Section~\ref{sec:triang-reduce} at the very end.

Note that due to the identities involving $D$, $T$, $\Phi_0$ and $E$,
the operator $D$ remains a morphism from $\D_D$ to each of the $E$,
$E$-$T$ and $E$-$T$-$\Phi_0$ equations. We will proceed under the
hypothesis that in that last case \emph{the morphism $D$ is actually
on-shell bijective}. That is, by an application of
Lemmas~\ref{lem:morphisms-maps} and~\ref{lem:morphism-induces-map} we
have the following mutually on-shell inverse morphisms
\begin{equation} \label{eq:E-T-Phi0-inv}
\begin{tikzcd}[column sep=huge,row sep=huge]
	\bullet
		\ar[swap]{d}{\D_D}
		\ar{r}{D}
	\&
	\bullet
		\ar{d}{\begin{bmatrix} E \\ T \\ \Phi_0 \end{bmatrix}}
	\\
	\bullet
		\ar[swap]{r}{\begin{bmatrix} T^* \\ R \\ 0 \end{bmatrix}}
	\&
	\bullet
\end{tikzcd}
	\quad \iff \quad
\begin{tikzcd}[column sep=huge,row sep=huge]
	\bullet
		\ar[swap]{d}{\begin{bmatrix} E \\ T \\ \Phi_0 \end{bmatrix}}
		\ar{r}{\bar{D}_0}
	\&
	\bullet
		\ar{d}{\D_D}
	\\
	\bullet
		\ar[swap]{r}{\begin{bmatrix} \bar{T}_0^* & -\Delta_T & -\Delta_\Phi \end{bmatrix}}
	\&
	\bullet
\end{tikzcd} \, ,
\end{equation}
for some operators $\bar{T}_0^*$, $\Delta_T$ and $\Delta_\Phi$. It would
be convenient if the above diagrams actually described an equivalence up
to homotopy. In general that will not be true, but only due to the
rather restricted notion of \emph{up to homotopy} that we have adopted
in this work. On the other hand, we make another hypothesis that the inverse
relationship between the above cochain maps takes on the following
slightly more general form, which we have actually already seen in the
hypotheses of Lemma~\ref{lem:top-triang-inv}:
\begin{align}
	\bar{D}_0 \circ D
	&= \id - h_D \circ \D_D \, ,
	\\
	\begin{bmatrix}
		\bar{T}_0^* & -\Delta_T & -\Delta_\Phi
	\end{bmatrix}
	\circ
	\begin{bmatrix}
		T^* \\ R \\ 0
	\end{bmatrix}
	&= \id - \D_D \circ h_D \, ,
	\\
	D \circ \bar{D}_0
	&= \id -
		\begin{bmatrix} \bar{h}_E & \bar{T}_2 & \bar{\Phi}_2 \end{bmatrix}
		\circ
		\begin{bmatrix} E \\ T \\ \Phi_0 \end{bmatrix}
	\, , \\
	\begin{bmatrix}
		T^* \\ R \\ 0
	\end{bmatrix}
	\circ
	\begin{bmatrix}
		\bar{T}_0^* & -\Delta_T & -\Delta_\Phi
	\end{bmatrix}
	&= \begin{bmatrix}
		\id & 0 & 0 \\
		0 & \id & 0 \\
		0 & 0 & \id
	\end{bmatrix}
	-
	\begin{bmatrix}
		E \\ T \\ \Phi_0
	\end{bmatrix}
	\circ
	\begin{bmatrix} \bar{h}_E & \bar{T}_2 & \bar{\Phi}_2 \end{bmatrix}
	-
	\bar{n} \, ,
\end{align}
\begin{equation}
	\text{where} \quad
	\bar{n}
	\circ
	\begin{bmatrix}
		0 & 0 \\
		0 & -\id \\
		-\id & 0
	\end{bmatrix}
	=
	\begin{bmatrix}
		\Phi_2^*    & \bar{D}_2^*   \\
    \bar{m}_{T,\Phi}    & \bar{m}_{T,T}   \\
    \bar{m}_{\Phi,\Phi} & \bar{m}_{\Phi,T}
	\end{bmatrix}
	\begin{bmatrix}
		\D_\Phi & \bar{\Delta}_\Phi \\
		0 & \D_D^*
	\end{bmatrix}
	\, ,
\end{equation}
for some choice of the newly introduced operators that satisfy the above
identities. Note that, as in the hypotheses of
Lemma~\ref{lem:top-triang-inv}, we require the operator $\bar{n}$ to be
a Noether identity for the $E$-$T$-$\Phi_0$ system that also satisfies
the factorization identity given by the last equation. 

Next applying Lemma~\ref{lem:top-triang-inv} we lift the above
equivalence first from the $E$-$T$-$\Phi_0$ system to the
block-triangular system~\eqref{eq:E-T-triang},
\begin{equation}
\begin{tikzcd}[column sep=huge,row sep=huge]
	\bullet
		\ar[swap]{d}{\begin{bmatrix}
			\D_D & \Delta_\Phi \\
			0 & \D_\Phi
		\end{bmatrix}}
		\ar{r}{\begin{bmatrix}
			D & \bar{\Phi}_2 \\
			0 & \id
		\end{bmatrix}}
	\&
	\bullet
		\ar{d}{\begin{bmatrix}
			E & 0 \\
			T & 0 \\
			\Phi_0 & -\id \\
			0 & \D_\Phi
		\end{bmatrix}}
	\\
	\bullet
		\ar[swap]{r}{\begin{bmatrix}
			T^* & \Phi_2^* \\
			R & \bar{m}_{T,\Phi} \\
			0 & \bar{m}_{\Phi,\Phi} \\
			0 & \id
		\end{bmatrix}}
	\&
	\bullet
\end{tikzcd}
	\iff
\begin{tikzcd}[column sep=huge,row sep=huge]
	\bullet
		\ar[swap]{d}{\begin{bmatrix}
			E & 0 \\
			T & 0 \\
			\Phi_0 & -\id \\
			0 & \D_\Phi
		\end{bmatrix}}
		\ar{r}{\begin{bmatrix}
			\bar{D}_0 & 0\\
			0 & \id
		\end{bmatrix}}
	\&
	\bullet
		\ar{d}{\begin{bmatrix}
			\D_D & \Delta_\Phi \\
			0 & \D_\Phi
		\end{bmatrix}}
	\\
	\bullet
		\ar[swap]{r}{\begin{bmatrix}
			\bar{T}_0^* & -\Delta_T & -\Delta_\Phi & 0 \\
			0 & 0 & 0 & \id
		\end{bmatrix}}
	\&
	\bullet
\end{tikzcd} \, ,
\end{equation}
and then to the block-triangular system~\eqref{eq:E-T-Phi0-triang},
\begin{equation}
\begin{tikzcd}[column sep=huge,row sep=huge]
	\bullet
		\ar[swap]{d}{\left[\begin{smallmatrix}
			\D_D & \Delta_\Phi & \Delta_T \\
			0 & \D_\Phi & \bar{\Delta}_\Phi \\
			0 & 0 & \D_D^*
		\end{smallmatrix}\right]}
		\ar{r}{\left[\begin{smallmatrix}
			D & \bar{\Phi}_2 & \bar{T}_2 \\
			0 & \id & 0 \\
			0 & 0 & \id
		\end{smallmatrix}\right]}
	\&
	\bullet
		\ar{d}{\left[\begin{smallmatrix}
			E & 0 & 0\\
			T & 0 & -\id \\
			\Phi_0 & -\id & 0 \\
			0 & \D_\Phi & \bar{\Delta}_\Phi \\
			0 & 0 & \D_D^*
		\end{smallmatrix}\right]}
	\\
	\bullet
		\ar[swap]{r}{\left[\begin{smallmatrix}
			T^* & \Phi_2^* & \bar{D}_2^* \\
			R & \bar{m}_{T,\Phi} & \bar{m}_{T,T} \\
			0 & \bar{m}_{\Phi,\Phi} & \bar{m}_{\Phi,T} \\
			0 & \id & 0 \\
			0 & 0 & \id
		\end{smallmatrix}\right]}
	\&
	\bullet
\end{tikzcd}
	\iff
\begin{tikzcd}[column sep=huge,row sep=huge]
	\bullet
		\ar[swap]{d}{\left[\begin{smallmatrix}
			E & 0 & 0\\
			T & 0 & -\id \\
			\Phi_0 & -\id & 0 \\
			0 & \D_\Phi & \bar{\Delta}_\Phi \\
			0 & 0 & \D_D^*
		\end{smallmatrix}\right]}
		\ar{r}{\left[\begin{smallmatrix}
			\bar{D}_0 & 0 & 0 \\
			0 & \id & 0 \\
			0 & 0 & \id
		\end{smallmatrix}\right]}
	\&
	\bullet
		\ar{d}{\left[\begin{smallmatrix}
			\D_D & \Delta_\Phi & \Delta_T \\
			0 & \D_\Phi & \bar{\Delta}_\Phi \\
			0 & 0 & \D_D^*
		\end{smallmatrix}\right]}
	\\
	\bullet
		\ar[swap]{r}{\left[\begin{smallmatrix}
			\bar{T}_0^* & -\Delta_T & -\Delta_\Phi & 0 & 0 \\
			0 & 0 & 0 & \id & 0 \\
			0 & 0 & 0 & 0 & \id
		\end{smallmatrix}\right]}
	\&
	\bullet
\end{tikzcd} \, .
\end{equation}

Next, to arrive at an equivalence up to homotopy with the original $E$
equation, we repeatedly apply Lemma~\ref{lem:top-triang-inv} and compose
the appropriate morphisms from~\eqref{eq:E-T-Phi0-triang}
and~\eqref{eq:E-T-triang} to first arrive at
\begin{equation}
\begin{tikzcd}[column sep=huge,row sep=huge]
	\bullet
		\ar[swap]{d}{\left[\begin{smallmatrix}
			\D_D & \Delta_\Phi & \Delta_T \\
			0 & \D_\Phi & \bar{\Delta}_\Phi \\
			0 & 0 & \D_D^*
		\end{smallmatrix}\right]}
		\ar{r}{\left[\begin{smallmatrix}
			D & \bar{\Phi}_2 & \bar{T}_2 \\
			0 & 0 & \id
		\end{smallmatrix}\right]}
	\&
	\bullet
		\ar{d}{\left[\begin{smallmatrix}
			E & 0\\
			T & -\id \\
			0 & \D_D^*
		\end{smallmatrix}\right]}
	\\
	\bullet
		\ar[swap]{r}{\left[\begin{smallmatrix}
			T^* & \Phi_2^* & \bar{D}_2^* \\
			R & \bar{m}_{T,\Phi} & \bar{m}_{T,T} \\
			0 & 0 & \id
		\end{smallmatrix}\right]}
	\&
	\bullet
\end{tikzcd}
	\iff
\begin{tikzcd}[column sep=huge,row sep=huge]
	\bullet
		\ar[swap]{d}{\left[\begin{smallmatrix}
			E & 0\\
			T & -\id \\
			0 & \D_D^*
		\end{smallmatrix}\right]}
		\ar{r}{\left[\begin{smallmatrix}
			\bar{D}_0 & 0 \\
			\Phi_0 & 0 \\
			0 & \id
		\end{smallmatrix}\right]}
	\&
	\bullet
		\ar{d}{\left[\begin{smallmatrix}
			\D_D & \Delta_\Phi & \Delta_T \\
			0 & \D_\Phi & \bar{\Delta}_\Phi \\
			0 & 0 & \D_D^*
		\end{smallmatrix}\right]}
	\\
	\bullet
		\ar[swap]{r}{\left[\begin{smallmatrix}
			\bar{T}_0^* & -\Delta_T & 0 \\
			\bar{\Phi}_0^* & -\bar{\Delta}_\Phi & 0 \\
			0 & 0 & \id
		\end{smallmatrix}\right]}
	\&
	\bullet
\end{tikzcd} \, ,
\end{equation}
and then finally arrive at
\begin{equation} \label{eq:triang-form}
\begin{tikzcd}[column sep=huge,row sep=huge]
	\bullet
		\ar[swap]{d}{E}
		\ar{r}{\begin{bmatrix}
			\bar{D}_0 \\
			\Phi_0 \\
			T
		\end{bmatrix}}
	\&
	\bullet
		\ar{d}{\begin{bmatrix}
			\D_D & \Delta_\Phi  & \Delta_T \\
			0 & \D_\Phi & \bar{\Delta}_\Phi  \\
			0 & 0 & \D_D^*
		\end{bmatrix}}
	\\
	\bullet
		\ar[swap]{r}{\begin{bmatrix}
			\bar{T}_0^* \\
			\bar{\Phi}_0^* \\
			D^*
		\end{bmatrix}}
	\&
	\bullet
\end{tikzcd}
	\quad \iff \quad
\begin{tikzcd}[column sep=huge,row sep=huge]
	\bullet
		\ar[swap]{d}{\begin{bmatrix}
			\D_D & \Delta_\Phi & \Delta_T \\
			0 & \D_\Phi & \bar{\Delta}_{\Phi} \\
			0 & 0 & \D_D^*
		\end{bmatrix}}
		\ar{r}{\begin{bmatrix}
			D & \bar{\Phi}_2 & \bar{T}_2
		\end{bmatrix}}
	\&
	\bullet
		\ar{d}{E}
	\\
	\bullet
		\ar[swap]{r}{\begin{bmatrix}
			T^* &
			\Phi_2^* &
			\bar{D}_2^*
		\end{bmatrix}}
	\&
	\bullet
\end{tikzcd} \, .
\end{equation}
By construction the above morphisms are mutually on-shell inverse and in
fact are an equivalence up to homotopy. The last claim would follow from
Lemma~\ref{lem:morphisms-maps}(d) when neither $E$ nor the triangular
decoupled system is over-determined (Definition~\ref{def:determined}).
In fact, in cases that are of interest to us, both systems will actually
be determined.

Next, using the idea from Section~\ref{sec:triang-reduce}, we can
reduce the triangular form~\eqref{eq:triang-form} further provided we
can find $\delta$ and $\eps$ operators satisfying the identities
\begin{align}
	\D_D \circ \delta_\Phi
	&= \Delta_\Phi + \eps_\Phi\circ \D_\Phi \, , \\
	\D_\Phi \circ \bar{\delta}_\Phi
	&= \bar{\Delta}_\Phi + \bar{\eps}_\Phi \circ \D_D^* \, , \\
	\D_D \circ \delta_T
	&= (\bar{\Delta}_T - \Delta) + \eps_T \circ \D_D^* \, ,
	\quad \text{where} \quad
	\bar{\Delta}_T = \Delta_T + \eps_\Phi\circ \bar{\Delta}_\Phi \, ,
\end{align}
for some operator $\Delta$. The result is the following simplified
triangular form:
\begin{multline} \label{eq:final-triang}
\begin{tikzcd}[column sep=huge,row sep=huge]
	\bullet
		\ar[swap]{d}{\begin{bmatrix}
			\D_D & \Delta_\Phi & \Delta_T \\
			0 & \D_\Phi & \bar{\Delta}_\Phi \\
			0 & 0 & \D_D^*
		\end{bmatrix}}
		\ar{r}{\begin{bmatrix}
			\id & \delta_\Phi & \delta_T \\
			0 & \id & \bar{\delta}_\Phi \\
			0 & 0 & \id
		\end{bmatrix}}
	\&
	\bullet
		\ar{d}{\begin{bmatrix}
			\D_D & 0 & \Delta \\
			0 & \D_\Phi & 0 \\
			0 & 0 & \D_D^*
		\end{bmatrix}}
	\\
	\bullet
		\ar[swap]{r}{\begin{bmatrix}
			\id & \eps_\Phi & \eps_T \\
			0 & \id & \bar{\eps}_\Phi \\
			0 & 0 & \id
		\end{bmatrix}}
	\&
	\bullet
\end{tikzcd}
	\\
	\iff
\begin{tikzcd}[column sep=huge,row sep=huge]
	\bullet
		\ar[swap]{d}{\begin{bmatrix}
			\D_D & 0 & \Delta \\
			0 & \D_\Phi & 0 \\
			0 & 0 & \D_D^*
		\end{bmatrix}}
		\ar{r}{\begin{bmatrix}
			\id & -\delta_\Phi & -\delta_T + \delta_\Phi\circ \bar{\delta}_\Phi \\
			0 & \id & -\bar{\delta}_\Phi \\
			0 & 0 & \id
		\end{bmatrix}}
	\&
	\bullet
		\ar{d}{\begin{bmatrix}
			\D_D & \Delta_\Phi & \Delta_T \\
			0 & \D_\Phi & \bar{\Delta}_\Phi \\
			0 & 0 & \D_D^*
		\end{bmatrix}}
	\\
	\bullet
		\ar[swap]{r}{\begin{bmatrix}
			\id & -\eps_\Phi & -\eps_T + \eps_\Phi\circ \bar{\eps}_\Phi \\
			0 & \id & -\bar{\eps}_\Phi \\
			0 & 0 & \id
		\end{bmatrix}}
	\&
	\bullet
\end{tikzcd}
\end{multline}

Thus, we arrive at the final form of the equivalence up to homotopy of
our original $E$ equation with the simplified triangular form
\begin{equation}
\begin{tikzcd}[column sep=huge,row sep=huge]
	\bullet
		\ar[swap]{d}{E}
		\ar{r}{\begin{bmatrix}
			\bar{D} \\
			\Phi \\
			T
		\end{bmatrix}}
	\&
	\bullet
		\ar{d}{\begin{bmatrix}
			\D_D & 0 & \Delta \\
			0 & \D_\Phi & 0 \\
			0 & 0 & \D_D^*
		\end{bmatrix}}
	\\
	\bullet
		\ar[swap]{r}{\begin{bmatrix}
			\bar{T}_1^* \\
			\bar{\Phi}_1^* \\
			D^*
		\end{bmatrix}}
	\&
	\bullet
\end{tikzcd}
	\quad \iff \quad
\begin{tikzcd}[column sep=huge,row sep=huge]
	\bullet
		\ar[swap]{d}{\begin{bmatrix}
			\D_D & 0 & \Delta \\
			0 & \D_\Phi & 0 \\
			0 & 0 & \D_D^*
		\end{bmatrix}}
		\ar{r}{\begin{bmatrix}
			D &
			\bar{\Phi} &
			\bar{T}
		\end{bmatrix}}
	\&
	\bullet
		\ar{d}{E}
	\\
	\bullet
		\ar[swap]{r}{\begin{bmatrix}
			T^* &
			\Phi_1^* &
			\bar{D}_1^*
		\end{bmatrix}}
	\&
	\bullet
\end{tikzcd} \, ,
\end{equation}
where
\begin{align}
	\begin{bmatrix}
		\bar{D} \\ \Phi \\ T
	\end{bmatrix}
	&= \begin{bmatrix}
		\id & \delta_\Phi & \delta_T \\
		0 & \id & \bar{\delta}_\Phi \\
		0 & 0 & \id
	\end{bmatrix}
	\circ
	\begin{bmatrix}
		\bar{D}_0 \\ \Phi_0 \\ T
	\end{bmatrix} \, ,
	\\
	\begin{bmatrix}
		\bar{T}_1 \\ \bar{\Phi}_1^* \\ D^*
	\end{bmatrix}
	&= \begin{bmatrix}
		\id & \eps_\Phi & \eps_T \\
		0 & \id & \bar{\eps}_\Phi \\
		0 & 0 & \id
	\end{bmatrix}
	\circ
	\begin{bmatrix}
		\bar{T}_0^* \\ \bar{\Phi}_0^* \\ D^*
	\end{bmatrix} \, ,
	\\
	\begin{bmatrix}
		D & \bar{\Phi} & \bar{T}
	\end{bmatrix}
	&= \begin{bmatrix}
		D & \bar{\Phi}_2 & \bar{T}_2
	\end{bmatrix}
	\circ
	\begin{bmatrix}
		\id & -\delta_\Phi & -\delta_T + \delta_\Phi\circ \bar{\delta}_\Phi \\
		0 & \id & -\bar{\delta}_\Phi \\
		0 & 0 & \id
	\end{bmatrix} \, ,
	\\
	\begin{bmatrix}
		T^* & \Phi_1^* & \bar{D}_1^*
	\end{bmatrix}
	&= \begin{bmatrix}
		T^* & \Phi_2^* & \bar{D}_2^*
	\end{bmatrix}
	\circ
	\begin{bmatrix}
		\id & -\eps_\Phi & -\eps_T + \eps_\Phi\circ \bar{\eps}_\Phi \\
		0 & \id & -\bar{\eps}_\Phi \\
		0 & 0 & \id
	\end{bmatrix} \, ,
\end{align}

In the next section, we apply the above strategy to an explicit example.

\section{Vector wave equation on Schwarzschild}\label{sec:vw}

Consider the exterior Schwarzschild spacetime $(\mathcal{M},\gf)$ of mass
$M>0$, where $\mathcal{M} \cong \mathbb{R}^2\times S_2$. If
$(S^2,\Omega)$ is the unit round sphere and $(t,r)$, with $-\oo<t<\oo$
and $2M<r<\oo$, are coordinates on the $\mathbb{R}^2$ factor, then the
Schwarzschild metric is
\begin{equation}
	\gf := g + r^2 \Omega \, ,
	\quad
	g := -f(r) \, {\d t}^2 + \frac{{\d r}^2}{f(r)} \, ,
	\quad
	f := 1 - \frac{2M}{r} \, .
\end{equation}
For convenience we also define
\begin{equation}
	f_1 := r f' = \frac{2M}{r} \, .
\end{equation}
Following the $2+2$ formalism of~\cite{martel-poisson}, spacetime tensors
decompose into spherical and radio-temporal sectors according to
\begin{equation}
	v_\mu \to \begin{pmatrix}
		v_a \\
		r u_A
	\end{pmatrix} \, ,
	\quad
	v^\nu \mapsto \begin{pmatrix}
		v^b \\
		\frac{1}{r} u^B
	\end{pmatrix} \, ,
\end{equation}
which is consistent both with raising-lowering $v_\mu \leftrightarrow
v^\mu$ with $\gf_{\mu\nu}$, as well as raising-lowering $v_a
\leftrightarrow v^a$, $u_A \leftrightarrow u^A$ with $g_{ab}$,
$\Omega_{AB}$, respectively.
If $D_A$, $\epsilon_{AB}$ and $Y^{lm}$ are respectively the Levi-Civita
connection, the volume form and the unit-normalized spherical harmonics
on $(S^2,\Omega)$, the following tensors define a basis for the even
($Y$) and odd ($X$) vector harmonics:
\begin{equation}
	Y_A = D_A Y \, ,
	\quad
	X_A = \epsilon_{BA} D^B Y \, .
\end{equation}
The tensor spherical harmonic decomposition of a covariant vector field
is then
\begin{align}
	v_\mu &= v_\mu^{\textrm{even}} + v_\mu^{\textrm{odd}}, \\
	v_\mu^{\textrm{even}}
		&\to \sum_{lm}
			\begin{pmatrix}
				v_a^{lm} Y^{lm} \\
				r u^{lm} Y_A^{lm}
			\end{pmatrix} , \\
	v_\mu^{\textrm{odd}}
		&\to \sum_{lm}
			\begin{pmatrix}
				0 \\
				r w^{lm} X_A^{lm}
			\end{pmatrix} .
\end{align}
For conciseness, we can omit the $lm$ indices from now on.

If $\grf_\mu$ is the spacetime Levi-Civita connection on
$(\mathcal{M},\gf)$, the vector wave equation is
\begin{equation}
	\dalf v_\mu := \grf_\nu \grf^\nu v_\mu = 0 \, .
\end{equation}

A separated mode
\begin{equation}
	v_\mu \to \begin{pmatrix}
		v_t (\d t)_a Y +  v_r (\d r)_a Y \\
		u\, r Y_A
	\end{pmatrix}
	e^{-i\omega t}
	+
	\begin{pmatrix}
		0 \\
		w\, r X_A
	\end{pmatrix}
	e^{-i\omega t} \, ,
\end{equation}
where $v_t = v_t(r)$, $v_r = v_r(r)$, $u = u(r)$ and $w = w(r)$,
satisfies the vector wave equation when
\begin{multline}
\label{eq:vwe-coord}
	\sqone_e 
	\begin{bmatrix}
		v_t \\ v_r \\ u
	\end{bmatrix}
	:=
	\begin{bmatrix}
		-\del_r \frac{1}{f} r^2 f\del_r v_t \\
		\del_r f r^2 f \del_r v_r \\
		\del_r \B_l r^2 f\del_r u
	\end{bmatrix}
	+ \left( \frac{\omega^2}{f} - \frac{\B_l}{r^2} \right)
	\begin{bmatrix}
		-\frac{1}{f} r^2 v_t \\
		f r^2 v_r \\
		\B_l r^2  u
	\end{bmatrix}
	\\
	+ i\omega r \frac{f_1}{f}
	\begin{bmatrix}
		 0 & 1 & 0 \\
		-1 & 0 & 0 \\
		 0 & 0 & 0
	\end{bmatrix}
	\begin{bmatrix}
		v_t \\ v_r \\ u
	\end{bmatrix}
	+
	\begin{bmatrix}
		0 & 0 & 0 \\
		0 & -2 f^2 & 2\B_l f \\
		0 & 2\B_l f & \B_l f_1
	\end{bmatrix}
	\begin{bmatrix}
		v_t \\ v_r \\ u
	\end{bmatrix}
	= 0 \, ,
\end{multline}

\begin{equation}\label{eq:vwo-coord}
	\sqone_o w :=
	\del_r \B_l r^2 f \del_r w
	+ \left(\frac{\omega^2}{f} - \frac{\B_l}{r^2}\right) \B_l r^2 w
	+ \B_l f_1 w
	= 0 \, ,
\end{equation}
where for convenience we have defined
\begin{equation}
	\B_l = l(l+1) \, .
\end{equation}
These are our \emph{radial mode equations}.
Note that we have written these equations in manifestly self-adjoint
form, $\sqone_e^* = \sqone_e$ and $\sqone_o^* = \sqone_e$
(cf.~Section~\ref{sec:adjoints}).

We may think of the vector wave equation as the harmonic gauge-fixed
Maxwell equations
\begin{equation}
	\sqf v_\mu = \grf^\nu \left(\grf_\nu v_\mu - \grf_\mu v_\nu\right)
		+ \grf_\mu \grf^\nu v_\nu \, .
\end{equation}
We will take the harmonic gauge condition and the gauge modes to
respectively be
\begin{equation}
	\psi_0 = r\, \grf^\mu v_\mu = 0 \, ,
	\quad
	v_\mu = \grf_\mu \frac{1}{r} \phi_0 \, .
\end{equation}

After the triangular decoupling, we will show that the vector wave
equation is equivalent to a system with \emph{spin-$s$ Regge-Wheeler
(RW)} operators $\D_s$~\cite[Eq.(23)]{rosa-dolan} on the diagonal, where
\begin{equation}
	\D_s \phi := \del_r f \del_r \phi
		- \frac{1}{r^2} [\B_l + (1-s^2)f_1] \phi + \frac{\omega^2}{f} \phi
	= 0 \, .
\end{equation}

In the following calculations, we will freely permit division by
$\omega$ and $\B_l = l(l+1)$. Hence our results will hold for the
$\omega \ne 0$, $l>0$ modes. A separate treatment would be required when
either of those conditions does not hold.

\subsection{Odd sector}
This sector is particularly simple, since it does not contain any gauge
modes and is not constrained by the harmonic gauge condition. The
transformation
\begin{equation}
	\phi_1 = -i\omega r w ,
	\quad
	w = -\frac{1}{i\omega r} \phi_1 ,
\end{equation}
puts $\sqone_o w = 0$ in direct equivalence with a $s=1$ RW equation:
\begin{equation} \label{eq:vwo-final-triang}
	\sqone_o w = \B_l i\omega r \frac{1}{\omega^2} \D_1 \phi_1 ,
	\quad
	\frac{1}{\omega^2} \D_1 \phi_1 = \frac{1}{\B_l i\omega r} \sqone_o w .
\end{equation}
In diagrammatic form, we have
\begin{equation}
\begin{tikzcd}[column sep=huge,row sep=huge]
	\bullet
		\ar[swap]{d}{\sqone_o}
		\ar[shift left]{r}{-i\omega r}
		\&
	\bullet
		\ar[shift left]{l}{-\frac{1}{i\omega r}}
		\ar{d}{\frac{1}{\omega^2} \D_1}
		\\
	\bullet
		\ar[shift left]{r}{\frac{1}{\B_l i\omega r}}
		\&
	\bullet
		\ar[shift left]{l}{\B_l i\omega r}
\end{tikzcd} \, ,
	\quad
	\begin{aligned}
		\left(-\frac{1}{i\omega r}\right) (-i\omega r) &= \id \, , \\
		\left(\frac{1}{\B_l i\omega r}\right)
			\left(\B_l i\omega r\right) &= \id \, .
	\end{aligned}
\end{equation}

\subsection{Even sector}
This sector is more complicated and is the first prototype for the
abstract approach to triangular decoupling that we outlined earlier in
Section~\ref{sec:decoupling}. Let us now explicitly introduce some of
the auxiliary operators that we will need, following the notation
established earlier:
\begin{align*}
	E &:= \sqone_e \, , \\
	D &:= \frac{1}{\omega^2 r^2}
		\begin{bmatrix}
		-i\omega r \\
		r^2 \del_r \frac{1}{r} \\
		1
		\end{bmatrix} \, , &
	T^* &:= \begin{bmatrix}
		\frac{i\omega r}{f} \\
		f r^2 \del_r \frac{1}{r} \\
		\B_l
	\end{bmatrix} \, , \\
	T &:= \begin{bmatrix}
		-\frac{i\omega r}{f} &
		-\frac{1}{r} \del_r f r^2 &
		\B_l
		\end{bmatrix} \, , &
	D^* &:= \frac{1}{\omega^2 r^2}
		\begin{bmatrix}
		i\omega r &
		-r\del_r &
		1
		\end{bmatrix} \, , \\
	\D_D &:= \frac{1}{\omega^2} \D_0 \, , &
	\D_D^* &:= \frac{1}{\omega^2} \D_0 \, , \\
	\Phi_0 &:= \begin{bmatrix}
		0 & -f & f\del_r r
		\end{bmatrix} \, , &
	\bar{\Phi}_0^* &:= \frac{1}{r^2}
		\begin{bmatrix}
			0 & -1 & \frac{1}{\B_l} r^2\del_r \frac{f}{r}
		\end{bmatrix} \, , \\
	\D_\Phi &:= \frac{1}{\omega^2} \D_1 \, , &
	R &= \bar{R} := -\id \, .
\end{align*}
They can be obtained from a mode decomposition of the spacetime
operators we discussed earlier, up to simple multiplicative
normalizations. We have chosen the normalizations so that all operators
have dimensionless components, under the convention that each of
$r^{-1}$, $\omega$, and $\del_r$ has the dimension of inverse length,
while $f$, $f_1$ and $\B_l$ are dimensionless. The operators
\begin{equation}
	\begin{bmatrix}
		v_t \\
		v_r \\
		u
	\end{bmatrix}
	= D[\phi_0] \, ,
	\quad
	\phi_1 = \Phi_0 \begin{bmatrix}
		v_t \\
		v_r \\
		u
	\end{bmatrix} \, ,
	\quad
	\psi_0 = T\begin{bmatrix}
		v_t \\
		v_r \\
		u
	\end{bmatrix}
\end{equation}
respectively separate out the \emph{gauge modes} ($\phi_0$), the
\emph{gauge invariant modes} ($\phi_1$) and the \emph{constraint
violating modes} ($\psi_0$). 

Computing the composition
\begin{equation*}
	\bar{\Delta}_\Phi := \bar{\Phi}_0^*\circ T^* \circ \bar{R}
	= -\frac{f_1}{\omega^2 r^2}
\end{equation*}
is enough to get us to the $E$-$T$-$\Phi_0$ triangular
form~\eqref{eq:E-T-Phi0-triang}.

Next, we must invert the morphism $D$ of gauge modes satisfying
$\D_D[\phi_0] = 0$ into the $E$-$T$-$\Phi_0$ system. At this step it is
helpful to use a computer algebra system. Any equivalence class of
scalar operators acting on the $v_t,v_r,u$ variables modulo the
$E$-$T$-$\Phi_0$ equations can be uniquely represented by a linear
combination of $v_t$ and $\del_r v_t$ (all second and higher derivatives
are eliminated by $E$, $\del_r v_r$ is eliminated by $T$ and $\del_r u$
is eliminated by $\Phi_0$). Precomposing each of them with $D$, we get
representatives of equivalence classes of scalar operators acting on the
$\phi_0$ variable modulo the $\D_D$ equation. These, in turn, are
uniquely represented by $\phi_0$ and $\del_r \phi_0$. Algebraically
inverting the relationship between $v_t, \del_r v_t$ and $\phi_0,\del_r
\phi_0$ gives us the operator $\bar{D}_0$
from~\eqref{eq:E-T-Phi0-triang}. Keeping track of how higher derivatives
are eliminated by the equations also allows us to complete $\bar{D}_0$
to a cochain map. The result is
\begin{align*}
	\bar{D}_0 &= \begin{bmatrix}
		i\omega r & 0 & 0
		\end{bmatrix} \, , \\
	\begin{bmatrix}
		\bar{T}_0^* & -\Delta_T & -\Delta_\Phi
	\end{bmatrix}
	&= \frac{1}{\omega^2 r^2} \begin{bmatrix}[ccc|c|c]
		-i\omega r f & -f_1 & 0 &
		-f_1 f r^2 \del_r \frac{1}{r} & 
		\B_l f_1
	\end{bmatrix} \, .
\end{align*}

Further explicit calculations give us the homotopy corrections needed to
proceed with the decoupling:
\begin{align}
	h_D &= 0 \, ,
	\\
	\begin{bmatrix}
		\bar{h}_E & \bar{T}_2 & \bar{\Phi}_2
	\end{bmatrix}
	&= \frac{1}{\omega^2 r^2} \begin{bmatrix}[ccc|c|c]
		0 & 0 & 0 &  0 & 0 \\
		0 & 1 & 0 &  fr^2\del_r\frac{1}{r} & -\B_l \\
		0 & 0 & \frac{f}{\B_l} &  f & -fr\del_r
	\end{bmatrix} \, ,
	\\
	\bar{n}
	&= \begin{bmatrix}
		0 & 0 \\
		-fr^2\del_r \frac{f}{r} & \B_l f \\
		-\B_l f & \B_l r\del_r f \\
		\cmidrule(lr){1-2}
		f & 0 \\
		\cmidrule(lr){1-2}
		0 & f
	\end{bmatrix}
	\begin{bmatrix}
		-D^* & \D_D^* & 0 \\
		-\bar{\Phi}_0^* & \bar{\Delta}_\Phi & \D_\Phi
	\end{bmatrix} \, ,
	\\
	\begin{bmatrix}
		\Phi_2^* & \bar{D}_2^* \\
		\bar{m}_{T,\Phi} & \bar{m}_{T,T} \\
		\bar{m}_{\Phi,\Phi} & \bar{m}_{\Phi,T} 
	\end{bmatrix}
	&= \begin{bmatrix}
		0 & 0 \\
		-\B_l f & fr^2\del_r \frac{f}{r} \\
		-\B_l r\del_r f & \B_l f \\
		\cmidrule(lr){1-2}
		0 & -f \\
		\cmidrule(lr){1-2}
		-f & 0
	\end{bmatrix} \, .
\end{align}

From the above operators, we can construct the desired triangular
decoupling~\eqref{eq:triang-form} of $E$:
\begin{align}
	\begin{bmatrix}
		\bar{D}_0 \\
		\Phi_0 \\
		T
	\end{bmatrix}
	&=
	\begin{bmatrix}
		i\omega r & 0 & 0 \\
		\cmidrule(lr){1-3}
		0 & -f & f\del_r r \\
		\cmidrule(lr){1-3}
		-\frac{i\omega r}{f} & -\frac{1}{r} \del_r f r^2 & \B_l
	\end{bmatrix} \, ,
	\\
	\begin{bmatrix}
		\bar{T}_0^* \\
		\bar{\Phi}_0^* \\
		D^*
	\end{bmatrix}
	&= \frac{1}{\omega^2 r^2} \begin{bmatrix}
		-i\omega r f & -f_1 & 0 \\
		\cmidrule(lr){1-3}
		0 & -1 & \frac{1}{\B_l} r^2 \del_r \frac{f}{r} \\
		\cmidrule(lr){1-3}
		i\omega r & -r\del_r & 1
	\end{bmatrix} \, ,
	\\
	\begin{bmatrix}
		D & \bar{\Phi}_2 & \bar{T}_2
	\end{bmatrix}
	&= \frac{1}{\omega^2 r^2} \begin{bmatrix}[c|c|c]
		-i\omega r & 0 & 0 \\
		r^2 \del_r \frac{1}{r} & -\B_l & fr^2\del_r\frac{1}{r} \\
		1 & -fr\del_r & f
	\end{bmatrix} \, ,
	\\
	\begin{bmatrix}
		T^* &
		\Phi_2^* &
		\bar{D}_2^*
	\end{bmatrix}
	&= \begin{bmatrix}[c|c|c]
		\frac{i\omega r}{f} & 0 & 0 \\
		f r^2 \del_r \frac{1}{r} & -\B_l f & fr^2\del_r \frac{f}{r} \\
		\B_l & -\B_l r\del_r f & \B_l f
	\end{bmatrix} \, .
\end{align}

An explicit calculation also shows that the cochain maps that we have
obtained are mutually inverse up to homotopy:
\begin{align}
	\begin{bmatrix}
		\bar{D}_0 \\
		\Phi_0 \\
		T
	\end{bmatrix}
	\circ
	\begin{bmatrix}
		D & \bar{\Phi}_2 & \bar{T}_2
	\end{bmatrix}
	&= \id - \begin{bmatrix}
		0 & 0 & 0 \\
		0 & f & 0 \\
		1 & 0 & f
	\end{bmatrix}
	\circ
	\begin{bmatrix}
		\D_D & \Delta_\Phi & \Delta_T \\
		0 & \D_\Phi & \bar{\Delta}_\Phi \\
		0 & 0 & \D_D^*
	\end{bmatrix} \, ,
	\\
	\begin{bmatrix}
		\bar{D}_0 \\
		\Phi_0 \\
		T
	\end{bmatrix}
	\circ
	\begin{bmatrix}
		D & \bar{\Phi}_2 & \bar{T}_2
	\end{bmatrix}
	&= \id - \frac{1}{\omega^2 r^2} \begin{bmatrix}
		0 & 0 & 0 \\
		0 & 1 & 0 \\
		0 & 0 & \frac{f}{\B_l}
	\end{bmatrix} \circ \sqone_e \, .
\end{align}

Noting the helpful identities
\begin{gather}
	\frac{1}{\omega^2} \D_{s_1} \frac{1}{(s_1^2-s_2^2)}
		= \frac{f_1}{\omega^2 r^2} + \frac{1}{(s_1^2-s_2^2)} \frac{1}{\omega^2} \D_{s_2} \, ,
	\\
	\frac{1}{\omega^2} \D_0 \left(\frac{f}{2}\right)
	= \left(\frac{f_1}{\omega^2 r^2} f r^2 \del_r \frac{1}{r}
			+ \frac{f_1^2}{2\omega^2 r^2}\right)
		+ \left(\frac{f}{2}\right) \frac{1}{\omega^2} \D_0
	\, ,
\end{gather}
we can use the idea of Section~\ref{sec:triang-reduce} to further reduce
the triangular form with the following operators:
\begin{align*}
	\bar{\delta}_\Phi &= -\id \, , &
		\delta_\Phi &= \B_l \, , &
		\delta_T &= \frac{f}{2} \, , \\
	\bar{\eps}_\Phi &= -\id \, , &
		\eps_\Phi &= \B_l \, , &
		\eps_T &= \frac{f}{2} \, ,
\end{align*}
which are needed to kill or simplify the off-diagonal terms
\begin{align}
\label{eq:vwe-bDelta-Phi}
	\bar{\Delta}_\Phi &= -\frac{f_1}{\omega^2 r^2} \, , \\
\label{eq:vwe-Delta-Phi}
	\Delta_\Phi &= -\frac{f_1}{\omega^2 r^2} \B_l \, , \\
\label{eq:vwe-Delta-T}
	\Delta_T &= \frac{f_1}{\omega^2 r^2} f r^2 \del_r \frac{1}{r} \, , \\
	\text{and} \quad
	\bar{\Delta}_T &= \Delta_T + \eps_\Phi \circ \bar{\Delta}_\Phi
	= \frac{f_1}{\omega^2 r^2}
		\left( f r^2 \del_r \frac{1}{r} - \B_l \right) \, .
\end{align}
The remaining off-diagonal term in the triangular
form~\eqref{eq:final-triang} is then
\begin{equation}
	\Delta
	= \bar{\Delta}_T
		- \left(\frac{f_1}{\omega^2 r^2} f r^2 \del_r \frac{1}{r}
			+ \frac{f_1^2}{2\omega^2 r^2}\right)
	= -\frac{f_1}{\omega^2 r^2} \left(\B_l + \frac{f_1}{2}\right) \, ,
\end{equation}
so that our final reduced triangular form is
\begin{equation} \label{eq:vwe-final-triang}
	\sqone_e
	\begin{bmatrix}
		v_t \\ v_r \\ u
	\end{bmatrix}
	= 0
	\quad \iff \quad
	\frac{1}{\omega^2} \begin{bmatrix}
		\D_0 & 0 & -\frac{f_1}{r^2} \left(\B_l + \frac{f_1}{2}\right) \\
		0 & \D_1 & 0 \\
		0 & 0 & \D_0
	\end{bmatrix}
	\begin{bmatrix}
		\phi_0 \\ \phi_1 \\ \psi_0
	\end{bmatrix}
	= 0
	\, .
\end{equation}

The following operators will provide us with the final pieces necessary
to construct a homotopy equivalence between $\sqone_e$ and the reduced
triangular form~\eqref{eq:vwe-final-triang}:
\begin{align}
	\begin{bmatrix}
		\phi_0 \\ \phi_1 \\ \psi_0
	\end{bmatrix}
	=
	\begin{bmatrix}
		\bar{D} \\ \Phi \\ T
	\end{bmatrix}
	\begin{bmatrix}
		v_t \\ v_r \\ u
	\end{bmatrix}
	&= \begin{bmatrix}
		1 & \B_l & \frac{f}{2} \\
		0 & 1 & -1 \\
		0 & 0 & 1
	\end{bmatrix}
	\circ
	\begin{bmatrix}
		i\omega r & 0 & 0 \\
		\cmidrule(lr){1-3}
		0 & -f & f\del_r r \\
		\cmidrule(lr){1-3}
		-\frac{i\omega r}{f} & -\frac{1}{r} \del_r f r^2 & \B_l
	\end{bmatrix}
	\begin{bmatrix}
		v_t \\ v_r \\ u
	\end{bmatrix}
	\, ,
	\\
	\begin{bmatrix}
		\bar{T}_1 \\ \bar{\Phi}_1^* \\ D^*
	\end{bmatrix}
	&= \begin{bmatrix}
		1 & \B_l & \frac{f}{2} \\
		0 & 1 & -1 \\
		0 & 0 & 1
	\end{bmatrix}
	\circ
	\frac{1}{\omega^2 r^2} \begin{bmatrix}
		-i\omega r f & -f_1 & 0 \\
		\cmidrule(lr){1-3}
		0 & -1 & \frac{1}{\B_l} r^2 \del_r \frac{f}{r} \\
		\cmidrule(lr){1-3}
		i\omega r & -r\del_r & 1
	\end{bmatrix} \, ,
	\\
	\begin{bmatrix}
		v_t \\ v_r \\ u
	\end{bmatrix}
	=
	\begin{bmatrix}
		D & \bar{\Phi} & \bar{T}
	\end{bmatrix}
	\begin{bmatrix}
		\phi_0 \\ \phi_1 \\ \psi_0
	\end{bmatrix}
	&=
	\frac{1}{\omega^2 r^2} \begin{bmatrix}[c|c|c]
		-i\omega r & 0 & 0 \\
		r^2 \del_r \frac{1}{r} & -\B_l & fr^2\del_r\frac{1}{r} \\
		1 & -fr\del_r & f
	\end{bmatrix}
	\circ
	\begin{bmatrix}
		1 & -\B_l & -\frac{f}{2}-\B_l \\
		0 & 1 & 1 \\
		0 & 0 & 1
	\end{bmatrix}
	\begin{bmatrix}
		\phi_0 \\ \phi_1 \\ \psi_0
	\end{bmatrix}
	\, ,
	\\
	\begin{bmatrix}
		T^* & \Phi_1^* & \bar{D}_1^*
	\end{bmatrix}
	&=
	\begin{bmatrix}[c|c|c]
		\frac{i\omega r}{f} & 0 & 0 \\
		f r^2 \del_r \frac{1}{r} & -\B_l f & fr^2\del_r \frac{f}{r} \\
		\B_l & -\B_l r\del_r f & \B_l f
	\end{bmatrix}
	\circ
	\begin{bmatrix}
		1 & -\B_l & -\frac{f}{2}-\B_l \\
		0 & 1 & 1 \\
		0 & 0 & 1
	\end{bmatrix} \, .
\end{align}

\section{Discussion}\label{sec:discussion}

We have shown how an abstract strategy, based on formal properties of
linear differential operators, can be used to decouple the radial mode
equations (with nonzero frequency $\omega$ and angular momentum $l$) of
the vector wave equation on Schwarzschild spacetime into a system of
hierarchically (that is \emph{triangularly}) coupled scalar equations.
Essentially, we have explicitly separated out the degrees of freedom of
the vector field into \emph{pure gauge}, \emph{gauge invariant} and
\emph{constraint violating} modes (with the interpretation that the
vector wave equation arose as the harmonic (or Lorenz) gauge fixed
Maxwell equations). This strategy was first successfully applied
in~\cite{berndtson}. However, it should be noted that this decoupling
strategy is never made explicit in~\cite{berndtson}, where it becomes
apparent only as a common pattern in the voluminous explicit
calculations in a sequence of examples of increasing complexity. One of
the main contributions of our work is to make this decoupling strategy
explicit and mathematically founded on basic ideas drawn from the theory
of $D$-modules~\cite[Sec.10.5]{seiler}, the category theoretic approach
to differential equations~\cite[Sec.VII.5]{vinogradov} and homological
algebra~\cite{weibel}.

It is worth pointing out that, even though our approach to triangular
decoupling and reduction of the triangular form owes much
to~\cite{berndtson}, the triangular form~\eqref{eq:vwe-final-triang}
appears in~\cite{berndtson} only implicitly. Moreover, the reduction of
the triangular form in~\cite{berndtson} only kills the
$\Delta_\Phi$~\eqref{eq:vwe-Delta-Phi} off-diagonal term. Since we have
also killed the $\bar{\Delta}_\Phi$~\eqref{eq:vwe-bDelta-Phi} and
simplified the $\Delta_T$~\eqref{eq:vwe-Delta-T} terms, our decoupling
goes a step further.

The next step is to apply the same decoupling strategy to linearized
gravity in harmonic (or de~Donder) gauge on the Schwarzschild spacetime.
In fact, this strategy was already successfully applied to that case
in~\cite{berndtson}, which was actually one of the main goals of that
work. Unfortunately, this decoupling for linearized gravity was achieved
only after some rather voluminous explicit calculations and the final
results were not presented in a very economic way. We believe that
following the abstract strategy in Section~\ref{sec:decoupling}, can
significantly streamline both the presentation of the final results and
the intermediate calculations. The details will be discussed elsewhere.

We believe that the same strategy, of triangular decoupling by
separating out \emph{purge gauge}, \emph{gauge invariant} and
\emph{constraint violating} modes should also work for harmonic gauge
field equations (either Maxwell or linearized gravity) on other
spacetime geometries where the physical degrees of freedom satisfy
decoupled equations. The Kerr spacetime is an example, where at the mode
level the gauge invariant degrees of freedom decouple and satisfy
Teukolsky or Fackerell-Ipser equations~\cite{chandrasekhar,
frolov-novikov}. The main uncertainty there is the degree to which the
triangular form can be explicitly reduced using the method of
Section~\ref{sec:triang-reduce}.

As was explained in the Introduction, the complicated form of the radial
mode equations~\eqref{eq:vwe-coord} and~\eqref{eq:vwo-coord} of the
vector wave equation is a significant obstacle to the study of their
analytical properties, and in particular their spectral theory. Once
these equations have been put into the decoupled triangular
form~\eqref{eq:vwe-final-triang} and~\eqref{eq:vwo-final-triang}, their
analytical properties can be more easily deduced from their diagonal
parts. These diagonal parts consist of spin-$s$ Regge-Wheeler equations,
whose analytical properties are well understood~\cite{chandrasekhar}.
Such information can then be used to construct both classical and
quantum propagators (or Green functions) for the vector wave equation on
Schwarzschild spacetime. The same will be true for linearized gravity,
once a similar decoupling has been achieved. For that purpose, the
$\omega = 0$ and $l=0$ modes would need to be investigated as well.
These applications will also be discussed elsewhere.

Finally, it would also be interesting to relate our decoupling strategy
to the \emph{adjoint operator method} of Wald~\cite{wald-adjoint,
ab-adjoint}. Recall that we have liberally used formal adjoints
(Section~\ref{sec:adjoints}) as part of our notation. This was
convenient because some of the operators that we used in the case of the
vector wave equation, in Section~\ref{sec:vw}, really were formal
adjoints of each other. However, we did not really use this property as
part of our decoupling strategy. On the other hand, supposing that we
can find a decoupling of a formally self-adjoint differential equation
$e=e^*$ into a triangular form $\bar{e}$, taking adjoints produces for
free another decoupling of $e$ but this time into $\bar{e}^*$. This is
illustrated in the following commutative diagrams:
\begin{equation}
\begin{tikzcd}[column sep=huge,row sep=huge]
	\bullet
		\ar[swap]{d}{e}
		\ar[shift left]{r}{k}
		\&
	\bullet
		\ar[shift left]{l}{\bar{k}}
		\ar{d}{\bar{e}}
		\\
	\bullet
		\ar[shift left]{r}{g}
		\&
	\bullet
		\ar[shift left]{l}{\bar{g}}
\end{tikzcd}
	\quad \iff \quad
\begin{tikzcd}[column sep=huge,row sep=huge]
	\bullet
		\ar[swap]{d}{e=e^*}
		\ar[shift left]{r}{\bar{g}^*}
		\&
	\bullet
		\ar[shift left]{l}{g^*}
		\ar{d}{\bar{e}^*}
		\\
	\bullet
		\ar[shift left]{r}{\bar{k}^*}
		\&
	\bullet
		\ar[shift left]{l}{k^*}
\end{tikzcd} \, .
\end{equation}
If $\bar{e}$ is in upper triangular form, then $\bar{e}^*$ is in lower
triangular form. In the vector wave equation
example~\eqref{eq:vwe-final-triang}, the adjoint equation $\bar{e}^*$ is
actually equivalent to the original equation $\bar{e}$ by the simple
interchange $\phi_0 \leftrightarrow \psi_0$. In such a case, composing
the equivalences of $e$ with $\bar{e}$, of $\bar{e}$ with $\bar{e}^*$,
and of $\bar{e}^*$ with $e^* = e$, we obtain a morphism from $e$ to
itself, or more precisely a pair of operators $\tilde{k}, \tilde{g}$
satisfying $e\circ \tilde{k} = \tilde{g} \circ e$. In the terminology
of~\cite{ab-adjoint}, $\tilde{k}$ is a \emph{symmetry operator} and it
maps solutions of $e$ to solutions of $e$ (in this case,
isomorphically). Thus, if $\tilde{k}$ is not simply proportional to a
constant, it gives an interesting way to generate new solutions from
known ones. The significance of the existence of such a symmetry
operator should be investigated further.

\section*{Acknowledgments}
The author was in part supported by the National Group of Mathematical
Physics (GNFM-INdAM). The author also thanks Francesco Bussola for
checking some of this paper's calculations as part of his MSc
thesis~\cite{bussola}.

\section*{References}

\bibliographystyle{iopart-num}
\bibliography{paper-grprop}

\end{document}